\documentclass[11pt]{article}
\usepackage[margin=3cm]{geometry}
 \usepackage{amsmath,amssymb,amsthm}
\usepackage[colorlinks=true,breaklinks=true,bookmarks=true,urlcolor=blue,
     citecolor=blue,linkcolor=blue,bookmarksopen=false,draft=false]{hyperref}

   \usepackage{color,framed}
\usepackage[font=small]{caption}
\usepackage{paralist}
\usepackage{float}
  \usepackage{etex}
 \usepackage{tikz}
  \usepackage{breakurl}

\usepackage{bm}

\tikzstyle{normalNodeS}=[circle, color=black!75!white, fill, draw, inner sep = 0.1em, minimum size = 1.5 em, scale=0.5]
\tikzstyle{labeledNodeS}=[circle, color=black!75!white, draw, inner sep = 0.1em, minimum size = 1.5em, scale=.75]
\tikzstyle{normalEdgeF}=[line width=1.6pt, color=black!75!white, >=stealth]
\tikzstyle{holdoverEdge}=[normalEdgeF, deeporange, thick]

\tikzstyle{demandNode}=[inner sep=0em]
\tikzstyle{demandEdge}=[>=stealth, double, thick, transparent]

\newcommand{\R}{\mathbb{R}} 

\renewcommand{\S}{\mathcal{S}}
\newcommand{\B}{\mathcal{B}}

\newcommand{\I}{{\mathcal{I}}}
\newcommand{\M}{{\mathcal{M}}}

\newcommand{\Delete}[1]{}

 \newtheorem{theorem}{Theorem}[section]
\newtheorem{proposition}[theorem]{Proposition}
\newtheorem{corollary}[theorem]{Corollary}

\newtheorem{lemma}[theorem]{Lemma}

\newtheorem{remark}[theorem]{Remark}
\newtheorem{definition}[theorem]{Definition}
\newtheorem{assumption}[theorem]{Assumption}
\newtheorem{example}[theorem]{Example}

\author{%
Satoru Fujishige\thanks{Research Institute for Mathematical Sciences, 
Kyoto University, Kyoto 606-8502, Japan.  E-mail: \texttt{fujishig@kurims.kyoto-u.ac.jp}.
 The present work is supported 
by JSPS Grant-in-Aid for Scientific Research (B) 25280004.} \and
Michel X. Goemans\thanks{Department of Mathematics, MIT, Cambridge, MA 02139, USA.  E-mail: \texttt{goemans@math.mit.edu}.} \and
  Tobias Harks\thanks{ Institute of Mathematics, University of Augsburg, 
  86135 Augsburg, Germany. Email: \texttt{tobias.harks@math.uni-augsburg.de}.
  } \and Britta Peis\thanks{School of Business and Economics, RWTH Aachen University, 52072 Aachen, Germany. Email: \texttt{peis@oms.rwth-aachen.de}.}
 \and Rico Zenklusen\thanks{Department of Mathematics, ETH Zurich, Zurich, Switzerland, and Department
of Applied Mathematics and Statistics, Johns Hopkins University, Baltimore, USA. Email: \texttt{ricoz@math.ethz.ch}.}
}

\title{Matroids are Immune to Braess Paradox}

\begin{document}

\maketitle
\begin{abstract}
The famous Braess paradox describes the counter-intuitive phenomenon in which, in certain settings, 
the increase  of resources, like building a new road within a congested network,
may in fact lead to larger costs for the players in an equilibrium.
In this paper, we consider general nonatomic congestion games and give a characterization of the combinatorial property of strategy spaces 
for which the Braess paradox does not occur.
In short, \emph{matroid bases} are precisely the required structure.
We prove this characterization by two novel sensitivity results
for convex separable optimization problems over polymatroid base polyhedra which may
be of independent interest.

\end{abstract}
\section{Introduction.}

In a congestion game (as introduced by Rosenthal~\cite{Rosenthal73a,Rosenthal73b}) there is a finite set of players that compete over a finite set of resources. A pure strategy of a player consists of a subset of resources, and the congestion cost of a resource depends only on the number of players choosing the same resource. 

\emph{Nonatomic} congestion games model the interaction of a large
number of players with the property
that the strategy choice of each player has only a negligible effect on the others.
In these kinds of model, it is usually assumed that there is a continuum of players 
partitioned into \emph{populations} and the strategy space available to a player of a population comprises a  population-specific set of allowable subsets of resources. 
A \emph{pure Nash equilibrium} of a nonatomic congestion game
is a strategy distribution from which
no player can unilaterally select a different subset of resources with strictly lower cost.
Here, the cost of a subset is simply defined as the sum of the resource costs.
Nonatomic congestion games have a wide range of applications, for example, they are used to model habitat selection in biology (cf.~Milchtaich~\cite{Milchtaich96b}),
queueing systems (cf.~Korilis et al.~\cite{Korilis99}) and packet routing in telecommunications (cf.~Qiu et al.~\cite{Qiu06}).
Perhaps the most famous example of a nonatomic congestion game appears
in the traffic model of Wardrop where the resources form a (directed) graph
and a population corresponds to a continuum of players that want
to travel from an origin to some destination in the graph. In this case,
the set of allowable subsets corresponds to the set of origin-destination
paths and the costs represent travel times. In a \emph{Wardrop} equilibrium (cf.~Wardrop~\cite{Wardrop52}) each player selects a path of minimum cost. 
The existence of Wardrop equilibria and their characterization via pure Nash equilibria of
an associated non-cooperative game (assuming continuity of cost functions) has been established since the early 50's, see Beckmann et al.~\cite{Beckmann56}. In~\cite{Beckmann56}, the authors
show that a strategy distribution is a Wardrop equilibrium if and only if it is a global minimum of an associated separable convex function known as the Beckmann potential.
This characterization further implies that for continuous and nondecreasing cost functions,
any Wardrop equilibrium has the same cost on every resource (cf.~{\rm \cite{CorreaS2011} and Remark~\ref{l.potential}}).

\subsection{Braess Paradox.}
In this paper we will study a well-known phenomenon 
originally discovered in the context of the Wardrop routing model:
Dietrich Braess, a German mathematician, published in 1968 a paper~\cite{Braess68} (see also
the paper~\cite{BraessNW05}) in which he showed
that adding a new arc to a transportation network might 
actually degrade the performance of the resulting Wardrop equilibrium.
Here  the performance is measured in terms of the total travel
time experienced by players in a Wardrop equilibrium. 
Let us briefly recall an example
of the Braess paradox. As depicted in Fig.~\ref{fig:braess},
there is a single-source single-destination network and 
we want to send one unit of flow from $s$ to $t$.
On the arcs, we indicate the travel cost per unit as a function of
the congestion; in particular, $1$ means that the travel cost per unit is one independent of the congestion, and $x$
signifies that the travel cost per unit is equal to the
congestion of the arc.
In the left network, the unique Wardrop flow sends evenly
  one-half units along the upper and lower path, respectively. This
  flow is also optimal having total cost of $3/2$.
  Suppose that a new fast road  is built (latency function is reduced from $\infty$ to $0$)
  connecting the two nodes in the middle, as shown in the right-hand-side figure.  
  The new (unique) Wardrop equilibrium sends its flow entirely along the zig-zag path
  having a total cost of $2$ and each player perceives a strictly larger path latency of $2$. 
  This example shows the paradoxical situation that 
  a network infrastructure improvement may actually hurt the resulting travel times of the new Wardrop equilibrium.
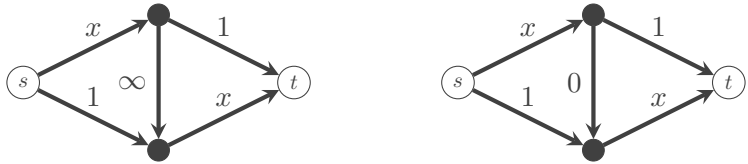
\begin{figure}[h!]
\begin{center}
\begin{tikzpicture}
 
 \begin{scope}[yshift=4.2cm]
 
  \node (fr) at (0, -0.4) {};
  
  \begin{scope}[scale=1.2]
 
 \draw (fr) +(0, -0.2) node[labeledNodeS] (s1) {$s$};
 \draw (s1) +(1.5, 0.75) node[normalNodeS] (v1) {}
 edge[normalEdgeF, <-] node[above] {$x$} (s1) ;
 \draw (s1) +(1.5, -0.75) node[normalNodeS] (v2) {}
 edge[normalEdgeF, <-] node[above] {$1$} (s1) 
 edge[normalEdgeF, <-] node[left] {$\infty$} (v1);
 \draw (s1) +(3, 0) node[labeledNodeS] (t1) {$t$}
 edge[normalEdgeF, <-] node[above] {$1$} (v1)
 edge[normalEdgeF, <-] node[above] {$x$} (v2)
;
 ;
 \end{scope}
 
 \end{scope}
\end{tikzpicture}
\begin{tikzpicture}
 
 \begin{scope}[yshift=4.2cm]

 \node (fr) at (0, -0.4) {};
 
 \draw (fr) ++(-0.5, 1.2) node (cal) {};
 
  \begin{scope}[scale=1.2]
 
 \draw (fr) +(1, -0.2) node[labeledNodeS] (s1) {$s$};
 \draw (s1) +(1.5, 0.75) node[normalNodeS] (v1) {}
 edge[normalEdgeF, <-] node[above] {$x$} (s1) ;
 \draw (s1) +(1.5, -0.75) node[normalNodeS] (v2) {}
 edge[normalEdgeF, <-] node[above] {$1$} (s1) 
 edge[normalEdgeF, <-] node[left] {$0$} (v1);
 \draw (s1) +(3, 0) node[labeledNodeS] (t1) {$t$}
 edge[normalEdgeF, <-] node[above] {$1$} (v1)
 edge[normalEdgeF, <-] node[above] {$x$} (v2)
;
 ;
 \end{scope}
 
 \end{scope}
\end{tikzpicture}
\end{center}
\caption[format=hang]{Example of the Braess paradox.}
\label{fig:braess}
\end{figure}

Let us now consider another type of Braess paradox that may arise
via demand reductions. Note that demand reductions frequently occur in practice, e.g., if commuters switch to the public transport system in case a new railway, tram or underground line has been built.

Consider the example in Fig.~\ref{fig:braess-demand}.
There are three populations $N=\{1,2,3\}$
that want to travel from $s_i$ to $t_i$, for $i=1,2,3$.
In the original instance the demands are $d_1=1,d_2=2$ and $d_3=M$.
The cost function $c(x)$ is defined by
$c(x)=0$ for $0\le x \le M$ and $c(x)=x-M$ for $M\le x$.
The resulting unique Wardrop equilibrium $x^*$ routes the flow
of population $1$ along the direct edge $(s_1,t_1)$. Thus, the total
cost of $x^*$ can be calculated as
$C(x^*)=1\cdot 2 + 2\cdot 2+M\cdot 0=6$. 
Suppose we decrease the demand of population $2$ from $2$ to $d_2=0$.
In the new (unique) Wardrop equilibrium $\bar x$, the flow of population $1$ will be sent
entirely on the path $(s_1,t_2,t_1)$ with a total cost of $C(\bar x)=M+2$.
It follows that for $M>4$, the reduction of demand may
actually hurt the total cost.

\begin{figure}[t!]
\begin{center}
\begin{tikzpicture}
 
 \begin{scope}[yshift=4.2cm]
 
  \node (fr) at (0, -0.4) {};
   \draw (fr) +(0, -0.2) node[labeledNodeS] (s1) {$s_1,s_2$};
 \draw (s1) +(2.5, 0.75) node[labeledNodeS] (v1) {$t_2,s_3$}
 edge[normalEdgeF, <-] node[above] {$x$} (s1) ;
 \draw (s1) +(2.5, -1.75) node[labeledNodeS] (v2) {$t_1,t_3$}
 edge[normalEdgeF, <-] node[above] {$2$} (s1) 
 edge[normalEdgeF, <-] node[right] {$c(x)$} (v1);
 ;
 \end{scope}
\end{tikzpicture}
\hspace{1cm}
 \begin{tikzpicture}[domain=2:4]
  \draw[->] (0,0.375) -- (4.2,0.375) node[right] {$x$};
  \draw[->] (0,0.375) -- (0,3.2);
  \draw[-,very thick] (0,0.375) -- (2.0,0.375) ;
  \draw[-,very thick] (2.0,0.375) -- (3.5,2.45) ;
   \draw[-,dotted] (0,1.75) -- (3,1.75) ;
    \draw[-,dotted] (3,1.75) -- (3,0.375) ;
    \node[text width=2cm] at (0.2,3) {$c(x)$};
     \node[text width=2cm] at (2.8,0.1) {$M$};
     \node[text width=2cm] at (3.7,0.1) {$M+1$};
      \node[text width=2cm] at (0.75,1.75) {$1$};
\end{tikzpicture}
\end{center}
\caption[format=hang]{Example of Braess paradox where a demand reduction
hurts the equilibrium cost. } \label{fig:braess-demand}
\end{figure}
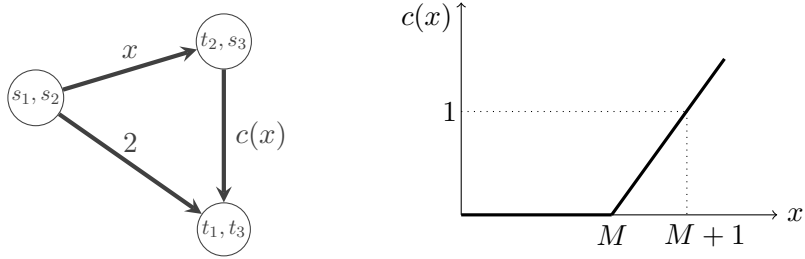

\subsection{Our Results and Techniques.}
We study nonatomic congestion games and
investigate the Braess paradox for arbitrary set systems and for both cost reductions and demand reductions
as explained in Fig.~\ref{fig:braess} and Fig.~\ref{fig:braess-demand}.
Note that there are interesting combinatorial structures of
the allowable subsets beyond paths in a graph: tours (as in the traveling salesman problem),
spanning trees, or Steiner trees (that frequently occur in telecommunication networks).

We differentiate between a \emph{weak} and a \emph{strong} form
of the Braess paradox. Weak Braess paradox occurs, if for the new equilibrium (after
cost and/or demand reductions), there exists
a resource with strictly increased cost. For the strong Braess paradox, there
must exist a player with strictly increased private cost. Note that the strong
Braess paradox implies the weak, and, immunity to
the weak Braess paradox implies immunity to the strong Braess paradox 
(but neither statement holds vice versa in general). 

\paragraph{Weak Braess Paradox.}
 Our first goal is to characterize 
the set of allowable subsets of players so that 
there will be no weak Braess paradox, no matter what kind of continuous and nondecreasing
cost functions are associated with the resources. An informal description of our main result is:
\begin{framed}
 A family of set systems is immune to the weak Braess paradox if and only if for
every set system of the family, the corresponding clutter (i.e., the set system containing only the inclusion-wise minimal sets) consists of bases of a population-specific \emph{matroid} defined on the ground set of resources.
\end{framed}
 We note that matroids have a rich combinatorial structure and include, for instance, the class of games, where each player wants to allocate a spanning tree in a graph.

Technically, our first characterization rests on two new results on the sensitivity
of optimal solutions minimizing 
a continuous, differentiable, nondecreasing and convex separable function (i.e., the Beckman potential) 
over a polymatroid base polytope. We show (cf.~Lemma~\ref{lem:cost-sen}) that if cost functions are shifted downwards, the new global minimum has the property that cost values
evaluated at a new optimal solution only decrease.
The second sensitivity result considers demand reductions which, as
we will argue, can be interpreted in terms of a decomposition
of a polymatroid.
More precisely, for the second sensitivity
result (cf.~Lemma~\ref{lem:dChangePolymat}) we consider a
specific polymatroid base polytope that can be decomposed as a Minkowski sum of
a finite number of polymatroid base polytopes. 
We show that by removing one polymatroid base polytope, any new optimal solution of the Beckmann potential has also decreased cost values. 
The connection of these two results to the Braess paradox 
is drawn by observing that for games with matroid structure, the problem of computing
a Wardrop equilibrium can be reduced to finding a global minimum of the Beckmann potential
over a sum of population-specific polymatroid base polytopes. For this we use the fact that the rank function of a matroid is a submodular function. The two sensitivity results, thus, imply
that for matroid set systems there will be no weak Braess paradox
no matter what kind of cost and/or demands reductions occur.
We prove the ``only if'' direction via exploiting the edge-vector 
characterization of base polytopes due to Tomizawa 
(see \cite[Theorem 17.1]{fujishige2005submodular}).

\paragraph{Strong Braess Paradox.}
Our second result gives a characterization of the occurrence of the strong 
Braess paradox. For this characterization we require that there
is no a priori description on how the individual
strategy spaces of populations interweave.\footnote{A formal definition 
of embeddings of strategy spaces into the resources is given 
in Section~\ref{sec:strongBP}; see Def.~\ref{def:embeddings}.} 
We say that a set system is \emph{universally} immune to
the strong Braess paradox, if it is immune to the strong Braess paradox no matter
how the strategy spaces of populations interweave.
We then obtain:
\begin{framed}
 A family of non-empty set systems containing at least two set systems is universally immune to the strong Braess paradox if and only if
for
every set system of the family, the corresponding clutter consists of bases of a population-specific \emph{matroid} defined on the ground set of resources.
\end{framed}
The ``if'' direction follows directly from our first characterization.
For the ``only if" direction we proceed by contradiction.  
If for a game with at least two populations, there exists a population with non-matroidal set system, then we derive appropriate cost functions on the resources, demands and an embedding of the
strategy spaces into resources such that the resulting game admits the strong Braess paradox.



\subsection{Related Literature.}
The discovery
of the Braess paradox has driven  a considerable amount of literature in different
fields of science ranging from transportation and traffic networks (cf.~\cite{Catoni91,Dafermos1984,frank81,Smith78}), queueing networks (cf.~\cite{CohenK90,Kameda02,Korilis99}), electrical and mechanical networks
(cf.~\cite{CohenH91}), computer science (cf.~\cite{CorSS04,FotakisKLS13,KamedaAKH00,LinRTW11,Rou02,Rough2002,Roughgarden02,Roug06,Roug06,Valiant10}) to economics (cf.~\cite{Poeppe92,Samuelson92}). For an overview of further works, we refer to the 
website maintained by Dietrich Braess~\cite{Braess-Web}.

In light of this substantial body of literature it seems surprising that to date
little is known regarding general  characterizations of 
the occurrence of the Braess paradox. Steinberg and Zangwill~\cite{Steinberg83}
and later Dafermos and Nagurney~\cite{Dafermos1984}, Pas and Principio~\cite{Pas1997}
and Hagstrom and Abrams~\cite{Hagstrom01}
derived \emph{instance-dependent}  necessary and sufficient conditions for the Braess paradox
to occur.
Here, instance-dependent means that these conditions depend on the concrete demand matrix, the cost functions and the network topology used. Hence, if for a given network topology the used cost functions or demand matrices are not known a priori, these works do not offer any insight on the occurrence of the Braess paradox. This situation occurs naturally 
whenever a network is build from scratch (as in telecommunications or mechanical networks)
or extended (as in traffic networks) and the traffic matrix and realized cost functions
are not known precisely. Even if the traffic matrix can be well estimated, the cost functions are subject to changes as street improvements and construction works are
continuously ongoing changing the street characteristics.
In such cases, it would be valuable to characterize networks 
that are not vulnerable to Braess paradox  for any instantiation of the demand matrix and the cost functions.
Milchtaich~\cite{Milchtaich06GEB} derived such a characterization by showing
that for undirected single $o$-$d$ networks, series-parallel graphs  form
the maximal graph class that is immune to the (strong) Braess paradox
no matter how many commuters travel and what kind of (continuous and nondecreasing)
cost function is used. Note that series-parallel networks are precisely the class of networks that do not contain the network in Fig.~\ref{fig:braess} as a topological minor.
He further proved that \emph{any} undirected single $o$-$d$ graph
that is not series-parallel can be equipped with carefully chosen costs and demands so that
the resulting instance admits the strong Braess paradox. His result, thus, provides a characterization
of undirected single $o$-$d$ graph topologies that are immune to the strong Braess paradox.
Very recently, Chen et al.~\cite{ChenDH15} and Cenciarelli et al.~\cite{Cenciarelli16} generalized the characterization of Milchtaich
towards directed graphs and allowing for multiple commodities.\footnote{There are related characterizations  of series-parallel graphs in different contexts, e.g., uniqueness and Pareto efficiency of Nash equilibria (cf.~\cite{Milchtaich05,Richman07}) and strong equilibria in congestion games
(cf.~\cite{Epstein09topo,Holzman97,HolzmanL03}).}

Some remarks are in order
to explain how our work differs from that of Milchtaich~\cite{Milchtaich06GEB}. As explained above, Milchtaich considered undirected single $o$-$d$ networks
and characterizes the maximal network topology that is immune to the strong Braess paradox. 
In particular, this implies that the
resources form an undirected graph, the strategy spaces of players are symmetric as the strategies
are the set of $o$-$d$ paths.
In contrast, we consider (general) nonatomic congestion games with \emph{asymmetric} strategy spaces, where for a player the allowable set of subsets of resources can have \emph{any} combinatorial
structure. Interesting cases beyond paths in graphs
include tours,
trees, or Steiner trees 
all in a
directed or undirected graph.  Additionally, we consider the more general
case of cost and/or demand reductions that might increase the equilibrium cost.

For our characterization of the strong Braess paradox, there is one important
additional difference to the result of Milchtaich. In contrast to Milchtaich's
characterization, we do not prescribe a priori how the sets of allowable subsets
of players are actually embedded in the ground set of resources, or said differently,
how the strategy spaces interweave. 


It is fair to say that matroids play a special role in the wide area of (integral) congestion games.
This connection was first discovered by Ackermann, R\"oglin and V\"ocking in the
important papers~\cite{Ackermann08,Ackermann09}.
In~\cite{Ackermann09}, they showed
that both weighted and player-specific
congestion games admit (pure Nash) equilibria in the case of matroid congestion
games, i.e., if the strategy space of each player consists of the bases of
a matroid on the set of resources. They also showed that the matroid property
is \emph{maximal} in the sense that whenever there are two players both having 
 allowable sets of resources that are not matroidal, then, there is a prescribed embedding of the sets
 into the ground set of resources and cost functions so that the resulting game does not have an equilibrium.  It should be noted that our characterization of the weak Braess paradox
is direct (relying on a polyhedral combinatorics point of view) and does not rely on the flexibility of embeddings. Also
 the ``only if'' direction of our characterization of the strong Braess paradox exhibits a difference to that used in~\cite{Ackermann09}. In~\cite{Ackermann09}
for obtaining counter examples it is required that the strategy space of \emph{all} players are non-matroidal,
whereas we only require that at least one player (or population in our setting) has a non-matroid set system, thus, allowing for a characterization.

Harks and Peis~\cite{HarksP14} considered a variant of congestion games, namely \emph{resource buying games},
in which players jointly design a resource infrastructure and share the congestion-dependent costs of the resources arbitrarily.
They showed that for marginally non-increasing cost functions such resource buying games always admit an equilibrium as long as the players' strategy spaces
form the base set of a matroid, while for non-matroid set systems, there is a two-player game with marginally non-increasing costs that does not admit an equilibrium. Finally, 
Harks et al.~\cite{harks2014resource} showed that integral-splittable congestion games with semi-convex cost functions always admit an equilibrium whenever each player's strategy space forms an integral polymatroid.

\section{Nonatomic Congestion Games.}\label{sec:model}
A tuple $\mathcal{M} = (N, E, (\mathcal{S}_i)_{i\in N}, (c_{e})_{e \in E},(d_i)_{i\in N})$ is called a
\emph{nonatomic congestion model} if $N = \{1,\dots,n\}$ is a non-empty, finite
set of populations and $E = \{e_1,\dots, e_m\}$ is a non-empty, finite set of
\emph{resources}.
Players are infinitesimally small, and each population $i$ consists of a continuum of players represented by the interval $[0,d_i]$ for some $d_i>0$.
For each population $i \in N$, the set
$\mathcal{S}_i$ is a non-empty, finite set of subsets
$S \subseteq E$ available to each player of population $i$. 
Each player of population $i\in N$ selects a strategy $S\in \S_i$, 
which leads to a strategy distribution 
$(x^i_S)_{S\in\S_i}$ satisfying $\sum_{S\in\S_i} x^i_S= d_i$ and $x^i_S\ge 0$ $(\forall S\in \S_i)$.
We denote by $\S$ the direct sum $\{(i,S)\mid i\in N, S\in \S_i\}$ of 
$\S_i$ $(i\in N)$, which represents 
the collection of all strategies of all players.
After each player of every population has chosen a strategy, we arrive at the overall strategy distribution $x=(x^i_S)_{(i,S)\in\S}$. 
The induced load of $x$ on $e$ is denoted by $x_e=\sum_{(i,S)\in \S:e\in S}x^i_S$ (assuming
every strategy $S$ of $(i,S)\in\S$ contains each resource at most once
and that for every $i\in N$, the rate of consumption of every $S\in\S_i$
on resource $e\in S$ is equal to one). Thus, we can compactly represent
the set of feasible strategy distributions by the following
polytope
\[ P(\mathcal{M}):=\left\{x\in \R_{\geq 0}^{\S} \;\middle\vert\;  \sum_{S\in\S_i} x^i_S= d_i \text{ for all }i\in N\right\},\]
where for $x\in \R_{\geq 0}^{\S}$ the value of $x$ on $(i,S)\in\S$ is denoted by $x^i_S$. 
We denote by 
$x_{i,e}=\sum_{S\in \S_i:e\in S}x^i_S$ the load of population $i$ on resource $e$.
Hence, $x_e = \sum_{i\in N} x_{i,e}$.
We impose the following assumption on cost functions.
\begin{assumption}\label{ass:costs}
For every resource $e \in E,$ 
we consider a \emph{cost function} $c_{e} : \R_{\geq 0} \rightarrow \R_{\geq 0}$
which is 
non-negative, continuous and nondecreasing. 
\end{assumption}
If in strategy distribution $x$, a player of population $i$ selects $S\in\S_i$, she perceives the disutility, or private cost, of 
\[ \pi_{i,S}(x)=\sum_{e\in S} c_{e}(x_e).\]

Since we are often interested in the load on the resources, we define
for every polytope $P(\mathcal{M})\subseteq \mathbb{R}_{\geq 0}^{\S}$
of feasible strategy distributions,
a corresponding polytope $\tilde{P}(\mathcal{M})\subseteq \mathbb{R}_{\geq 0}^{E}$
that captures all possible load vectors on the resources obtained by playing
a feasible strategy distribution, i.e.,
\begin{equation*}
\tilde{P}(\mathcal{M}) := \left\{\sum_{i\in N}\sum_{S\in \S_i} x^i_S\cdot \chi_{S} 
\;\middle\vert\;  x\in P(\mathcal{M})\right\},
\end{equation*}
where $\chi_S\in \{0,1\}^E$ for $S\subseteq E$ is the characteristic 
vector of $S$; hence,
$\chi_S(e)=1$ if $e\in S$ and $\chi_S(e)=0$ if $e\in E\setminus S$.
Note that defining  for each population $i\in N$ a polytope
\begin{equation*}
\tilde{P}_i(\mathcal{M}) := \left\{\sum_{S\in \S_i} x^i_S\cdot \chi_{S} 
\;\middle\vert\;  \sum_{S\in\S_i} x^i_S= d_i,\ x^i_S\ge 0\ (\forall S\in \S_i)
\right\},
\end{equation*}
we have $\tilde{P}(\mathcal{M})=\sum_{i\in N}\tilde{P}_i(\mathcal{M})$, 
the Minkowski sum of polytopes $\tilde{P}_i(\mathcal{M})$ $(i\in N)$.

\subsection{Nonatomic Matroid Congestion Games}\label{subsec:matroids}
A matroid is a tuple $M = (E,\I)$, where $E$ is a finite set, called the \emph{ground set}, and $\mathcal{I}\subseteq 2^E$ is a nonempty family of subsets of $E$, called \emph{independent sets}, such that: (i) if $X\in \mathcal{I}$ and $Y\subseteq X$, then $Y\in \mathcal{I}$, and (ii) if $X,Y\in \mathcal{I}$ with $|X|> |Y|$, then $\exists$ $e\in X\setminus Y$ such that $Y\cup \{e\}\in \mathcal{I}$.
The inclusionwise maximal independent sets of $\mathcal{I}$ are called \emph{bases} of matroid $M$, and usually denoted by $\B$, or $\B(M)$.
See~\cite{oxley1992matroid,Schrijver03,welsh2010matroid} for more information on matroids.

A nonatomic congestion model
$\mathcal{M}$ is called \emph{matroid congestion model} if for every
$i~\in~N$ there is a matroid $M_i = (E,\I_i)$ such that
$\mathcal{S}_i$ equals the set of bases of $M_i$. 
In case of nonatomic matroid congestion games we will write $\B_i$ instead of $\S_i$, $B_i$ instead of $S_i$ and $B$ instead of $S$.
We give three examples in the area of queueing, facility location and minimum spanning tree games.
\begin{example}[Queueing Games (cf.~\cite{Korilis99})]
There is a set $Q=\{q_1,\dots,q_m\}$ of $M/M/1$ queues served in a first-come-first-served fashion
and a set of $N=\{1,\dots,n\}$ independent Poisson arrivals of packets, where
the arrival rates are denoted by $d_1,\dots, d_n$. 
Every queue $q$ has a single server with exponentially distributed service time 
with mean $1/\mu_q$, $\mu_q>0$. Each packet is routed to a single queue $q$ out of a
set of allowable queues depending on the type. Given a distribution of
packets $x\in \R^m_{\geq 0}$, the mean delay of queue $q$ can be
computed as $c_{q}(x_q)=\frac{1}{\mu_{q}-x_q}.$
In this case, the sets
$\S_i$, $i\in N,$ are uniform rank-$1$ matroids. 
\end{example}

\begin{example}[Facility location games with supply functions]
Matroid congestion games can also be used as a modeling tool for resource buying games with
supply functions. More precisely, whereas so far we interpreted the cost of a resource
mostly in terms of a ``disutility'' like congestion, one can as well interpret costs
as actual renting or buying costs of a resource that depend on the demand, i.e., the more
users use a resource, the higher its price.
We provide an example phrased in the context of facility location, though the described
approach also applies to further settings. Consider a finite
set $E=\{e_1,\dots, e_m\}$ of resources in different locations, and a set of populations
$N=\{1,\dots, n\}$.
The resources could for example correspond to data centers, and the players
have to decide which data centers to use to serve their clients.
A population groups together players that want to serve clients
within the same areas. 
Each player in population $i\in N$---where, as usual, we assume
that there is a total ``mass'' of $d_i$ players in population $i$---desires
to use some number $k_i\in \mathbb{Z}_{> 0}$ of different data centers,
to cover $k_i$ different areas.
Each area $j$ can be served by any data center within a given set $S_j\subseteq E$.
The sets $S_j$ may overlap, even for the same player $i$. However, due to reliability
reasons, a player cannot use the same data center more than once.
Furthermore, to model an offer/demand interplay, the cost $c_e$ for using a particular
data center $e\in E$ depends on the total load of players who use data center $e$.
The higher the load on a data center, the larger the cost to use it.
In this setting, the strategy space of each population $i\in N$ corresponds
to a \emph{transversal matroid} described by the sets $S_j$ for all areas $j$ that
population $i$ wants to serve. 

\end{example}

\begin{example}[MST Games]
We are given an undirected graph $G=(V,E)$ with non-negative, continuous and non-decreasing edge cost functions $c_{e}(\ell), e\in E$. In a minimum spanning tree (MST) game, every population $i$
is associated with a demand interval $[0,d_i]$ and
a subgraph $G_i$ of $G$. A strategy distribution for population $i$ is to
route its demand along the spanning trees of $G_i$. Formally,  the edges correspond
to the resources and the sets
$\S_i$, $i\in N$, are the spanning trees of $G_i$. $M_i$ is called 
a \emph{graphic matroid}. 
\end{example}

\subsection{Wardrop Equilibria.}
A Wardrop equilibrium $x$ for a nonatomic congestion game
is a strategy distribution $x$ such that every player of every population
uses a strategy with minimum cost. Formally, 
\[ \pi_i(x):=\sum_{e\in S} c_{e}(x_e)\leq \sum_{e\in S'}c_{e}(x_e), 
\text{ for any } S,S'\in\S_i \text{ with }x^i_{S}>0, 
\text{ for all }  i\in N.\]
We recall the following characterization of Wardrop equilibria 
which implies their existence.
\begin{theorem}\label{thm:beck} {\rm (Beckmann et al.~\cite[Section 3.1.2]{Beckmann56})}
A strategy distribution $x$ is a Wardrop equilibrium if and only if it
is an optimal solution to
\begin{align}\label{beck:potential} \min_{x\in P(\M)}\left\{\Phi(x):=\sum_{e\in E} \int\limits_{0}^{x_e} c_e(t)\;dt\right\}. 
\end{align}
We call $\Phi$ the Beckmann potential.
\end{theorem}

Notice that the problem of finding the minimum value of the Beckmann potential can equivalently
be written in terms of $\tilde{P}(\mathcal{M})$ as the following minimization problem:
\begin{equation*}
\min_{x\in \tilde{P}(\mathcal{M})} \left\{\sum_{e\in E} \int_0^{x_e} c_e(t)\; dt\right\}.
\end{equation*}
Here it should be noted that $x=(x_e)_{e\in E}\in\mathbb{R}^E$, 
while $x$ appearing in (\ref{beck:potential}) is a strategy distribution 
in $\mathbb{R}^\S$.
Later in our results, we will often refer to this equivalent
version of the problem of minimizing the Beckmann potential.
For simplicity, we will use $\Phi(x)$ also for the Beckmann potential
for load vectors 
$x\in \tilde{P}(\mathcal{M})$.

\begin{remark}\label{l.potential}
Using that every Wardrop equilibrium  
$x\in {P}(\mathcal{M})$ is a global minimum of~\eqref{beck:potential}, 
we obtain the following well-known properties 
{\rm (}cf.~{\rm \cite{CorreaS2011})}. 
If cost functions $(c_e)_{e\in E}$ are strictly increasing, 
the Wardrop equilibrium load vector $(x_e)_{e\in E}\in \tilde{P}(\mathcal{M})$ is unique {\rm (}but there can be different 
decompositions of the demands among the subsets, that is,
there can be $x,y\in P(\mathcal{M}), x\neq y$ 
with $x_e=y_e$ for all $e\in E${\rm )}. For the case 
of nondecreasing costs, the vector of costs $(c_e(x_e))_{e\in E}$ is 
unique under the possibly non-unique equilibrium load vectors.
\end{remark}

\subsection{The Braess Paradox.}
Recall the examples
of Braess paradox presented in Fig.~\ref{fig:braess} and Fig.~\ref{fig:braess-demand}. In these examples,
the equilibrium flows on two (network) congestion models $\M$ and
$\bar \M$ are compared to each other, where $\bar \M$ is related to $\M$
by simply reducing some of the cost functions (in case of the example
in Fig.~\ref{fig:braess}, only one cost function is reduced from $\infty$ to $0$)
and/or reducing the demands of the populations.

In this work, we allow for general cost reductions of
the form $\bar c_e(t)\leq c_e(t)$ for all $t\geq 0$ and $e\in E$
and general demand reductions $\bar d_i\leq d_i, i\in N$.
We denote the changed model by $\bar{\mathcal{M}}$.
Note that for both models $\mathcal{M}$ and $\bar{\mathcal{M}}$,
the sets of allowable subsets $(\S_i)_{i\in N}$ remain the same.
We define the following notion of the weak and strong Braess paradox.

\begin{definition}[Weak/Strong Braess paradox]\label{def:braess}
Let $E=\{e_1,\dots,e_m\}$ be a finite set of resources.
A family of set systems $(E,\S_i)_{i\in N}$ with $\S_i\subseteq 2^E$ for all $i\in N$ admits the \emph{weak} Braess paradox, if there are two
nonatomic congestion models 
$\mathcal{M} = (N, E,(\S_i)_{i\in N}, (c_e)_{e \in E},(d_i)_{i\in N})$
and $\bar{\mathcal{M}} = (N, E, (\S_i)_{i\in N}, (\bar c_e)_{e \in E},(\bar d_i)_{i\in N})$,  
with
$\bar c_e(t)\leq c_e(t)$ for all $t\geq 0$ and $\bar d_i\leq d_i$ for all $i\in N$,
and two Wardrop
equilibria $x$ and $\bar x$ for $\M$ and $\bar{\mathcal{M}}$, respectively, 
such that 
\[\tag{\text{weak BP}}\text{there is\ \ $e\in E$ with } c_e(x_e)<\bar c_e(\bar x_e).\]
$(E,\S_i)_{i\in N}$ admits the \emph{strong} Braess paradox, if there is
\ $i\in N$ with $S,S'\in\S_i, x^i_S>0,{\bar x}^i_{S'}>0$ such that
\[\tag{\text{strong BP}} \pi_i(x)=\sum_{e\in S}c_e(x_e)<\sum_{e\in S'}\bar c_e(\bar x_e)=\pi_i(\bar x).\]
We say that $(E,\S_i)_{i\in N}$ is \emph{immune} to the weak/strong Braess paradox
if no such $\M,\bar \M$ exist.
\end{definition}

\begin{remark}
The strong Braess paradox implies the weak, but not vice versa.
On the other hand, if a family of set systems $(E,\S_i)_{i\in N}$ is immune to the weak Braess paradox,
it is also immune to the strong Braess paradox. Moreover, if $(E,\S_i)_{i\in N}$
satisfies that for every resource $e\in E$, there is a player $i$ with $\S_i=\{\{e\}\}$, then,
weak and strong Braess paradox are equivalent.
\end{remark}
The examples in Fig.~\ref{fig:braess} and Fig.~\ref{fig:braess-demand} already show that there are quite simple set systems $(E,\S_i)_{i\in N}$ that admit
the strong (and thus the weak) Braess paradox.

The driving question of this paper is the following:
\begin{framed}
What is a characterizing property of the set systems $(E,\S_i)_{i\in N}$ so that they 
are immune to the weak or strong Braess paradox?
\end{framed}

\section{A Characterization of the Weak Braess Paradox.}
For our first characterization of the weak Braess paradox, we
define for a set system $(E,\S_i)$
the \emph{clutter} 
as:
\[ (\S_i)^{\min}:=\{ U\in \S_i |\; \nexists \;T\in \S_i \text{ with }T\subset U\}.\]
The set system $(E,(\S_i)^{\min})$ (for which we also use the term ``clutter") 
contains all inclusion-wise minimal subsets of $\S_i$.

Our first main result gives  a complete characterization of the weak Braess paradox.
\begin{framed}
\begin{theorem}\label{thm:braess}
Let $(E,\S_i)_{i\in N}$ be a family of set systems, i.e., $\S_i\subseteq 2^E$ for all $i\in N$.
Then, the following statements are equivalent.
\begin{enumerate}
\item[{\rm (I)}]
 The clutter $(E,(\S_i)^{\min})$ forms the base set 
of a matroid $M_i=(E,\I_i)$ for all $i\in N.$
\item[{\rm (II)}]
$(E,\S_i)_{i\in N}$ is immune to the weak Braess paradox.
\end{enumerate}
\end{theorem}
\end{framed}

\subsection{Proof of Theorem \ref{thm:braess}: (I) $\Rightarrow$ (II).}
\noindent
The proof of (I) $\Rightarrow$ (II) consists of a number of steps 
organized as follows.

In the first step we prove that if every clutter $(E,(\S_i)^{\min})$
corresponds to bases of some matroid $M_i=(E,\I_i)$ for all $i\in N$, then  
$(E,(\S_i)^{\min})_{i\in N}$ is immune to the weak Braess paradox.
For showing this, we first model the set of feasible strategy distributions of a 
nonatomic matroid congestion game via a suitably defined \emph{polymatroid base polytope}.
This way,  the problem of computing
a Wardrop equilibrium of a matroid congestion model can be interpreted as the problem to find a global minimum of 
a separable convex function (i.e., the Beckman potential)
over a sum of population-specific polymatroid base polytopes which itself
is a polymatroid base polytope (see \cite{edmonds1970submodular}).

In the next steps, we prove two sensitivity results for this class of optimization problems stating that 
whenever (i) cost functions or (ii) demands are decreased, any global minimizer
of the Beckmann potential has the property that the new induced cost values 
component-wise decrease. This implies that the weak Braess paradox does not occur.

In the final step of (I) $\Rightarrow$ (II), we prove that if the 
family of clutters 
$(E,(\S_i)^{\min})_{i\in N}$ is immune to the weak Braess paradox, 
then so is the family of 
set systems $(E,\S_i)_{i\in N}$.

\paragraph{Polymatroids.}\label{par:polymatroid}
In order to define polymatroids we first have to introduce submodular functions.
A function $\rho:2^E \rightarrow \mathbb{R}$ is called submodular if $\rho(U)+\rho(V) \geq \rho(U \cup V) + \rho(U \cap V)$ for all $U, V \subseteq E$.
It is called monotone if $\rho(U) \leq \rho(V)$ for all $U \subseteq V\subseteq E$, and normalized if $\rho(\emptyset) = 0$.
Given a submodular, monotone and normalized function $\rho$, the pair $(E,\rho)$ is called a \emph{polymatroid}. The associated \emph{polymatroid base polytope} is defined as
\[ P_{\rho} := \left\{ x \in \mathbb{\R}_{+}^{E} \mid x(U) \leq \rho(U) 
\; \forall U \subseteq E, \; x(E) = \rho(E)\right\},\]
where $x(U) := \sum_{e \in U} x_e$ for all $U\subseteq E$.
Given submodular functions $\rho_i, i \in N$ all defined on $2^E$
and $\rho:=\sum_{i\in N}\rho_i$, we know that the Minkowski sum
 $P_\rho = \sum_{i \in N} P_{\rho_i}$
is also a polymatroid base polytope, see~\cite{edmonds1970submodular}, \cite{fujishige2005submodular}, or \cite[Theorem 44.6]{Schrijver03}.

\paragraph{From Nonatomic Matroid Congestion Games to Polymatroids.}\label{par:non-polymatroid}
Consider now a nonatomic matroid congestion model $\M$, where
for every $i~\in~N$ the associated strategy space forms the base set $\B_i$ of  a matroid $M_i = (E,\I_i)$.
It is well-known that the rank function $\text{rk}_i:2^E\to \R$ of 
matroid $M_i$ satisfies
$$\text{rk}_i(S):=\max\{|U| \mid U\subseteq S \text{ and } U\in \I_i\}\quad 
\forall S\subseteq E$$
and is submodular, monotone and normalized. 
Moreover, the characteristic vectors of the bases in $\B_i$ are exactly the vertices of the polymatroid base polytope
$P_{\text{rk}_i}$.

It follows that the polytope
\[ \left\{x^i\in \R^{\B_i}_+ \;\middle\vert\; \sum_{B\in \B_i} x_B^i =d_i\right\}\]
corresponds to strategy distributions for population $i$ that lead to load vectors
in the following polytope:
\[ P_{d_i\cdot\text{rk}_i} = \left\{ x_i \in \mathbb{\R}_{+}^{E} \mid x_i(U) \leq d_i\cdot \text{rk}_i(U) \; \forall U \subseteq E, \; x_i(E) = d_i\cdot \text{rk}_i(E) \right\}.\]
Thus, the polymatroid base polytope $P:=\sum_{i\in N}P_{d_i\cdot\text{rk}_i}=P_{\sum_{i\in N} d_i \cdot \text{rk}_i}$
is equal to $\tilde{P}(\M)$.
To simplify notation we define the following submodular
functions: $\rho_i = d_i\cdot \mathrm{rk}_i$ for $i\in N$
and $\rho = \sum_{i\in N} \rho_i$.
Furthermore, let $P_i= P_{\rho_i}$.
We thus have
$\tilde{P}(\mathcal{M})= P_\rho = \sum_{i\in N} P_i$.

\paragraph{Two Sensitivity Results.}\label{par:sens}
Consider the following optimization problem 
\begin{align}\label{potential} \min_{x\in P_\rho}\left\{\Phi(x):=\sum_{e\in E} \int\limits_{0}^{x_e} c_e(t)\;dt\right\},
\end{align}
where $P_\rho$ is a polymatroid base polytope with rank function $\rho$ 
and for all $e\in E$, $c_e:\R_{\geq 0}\rightarrow\R_{\geq 0},$ are 
non-decreasing and continuous functions. 
We recall the following necessary and sufficient optimality
conditions.
Let $\chi_e\in\R_{\geq 0}^{E}$ for $e\in E$ 
be the indicator vector with all-zero
entries except for the $e$-th coordinate which is $1$.
From Fujishige~\cite{fujishige2005submodular} we know that
a base $x\in P_\rho$ is optimal for problem~\eqref{potential} if and only if 
\begin{align}\label{thm:fuj}
c_e(x_e)\leq c_{f}(x_{f}) \text{ for any 
$e,f\in E$ such that  
$x':=x+\epsilon(\chi_f-\chi_e)\in P_\rho$
for some $\epsilon>0$}.
\end{align}

We now prove a result
on the sensitivity of optimal solutions minimizing the Beckmann potential over a polymatroid base polytope: we will show that
whenever a cost function is shifted downwards (i.e., $\bar c_e(t)\leq c_e(t)$ for all $t\geq 0$),
then, any new optimal solution $\bar x$ has the property $\bar c_e(\bar x_e)\leq c_e(x_e)$ for all $e\in E$,
where $x$ denotes any optimal solution for the cost functions $c_e, e\in E$, and
$\bar{x}$ denotes any optimal solution for the cost functions $\bar{c}_e, e\in E$.
 This result implies
that for matroid set systems the weak Braess paradox does not occur if only cost reductions are considered.
\begin{lemma}\label{lem:cost-sen}
Let $I$ and $\bar I$ be two instances of problem~\eqref{potential} with the only difference that
for $\bar I$ we use cost functions satisfying $\bar c_e(t)\leq c_e(t)$ for all $e\in E$ and $t\geq 0$.
Then, any optimal solutions $x$ and $\bar x$
to the instances $I$ and $\bar I$, respectively, satisfy
\[ \bar c_e(\bar x_e)\leq c_e(x_e), \text{ for all }e\in E.\]
\end{lemma}
\begin{proof}
Assume by contradiction that there is $e\in E$ with $\bar c_e(\bar x_e)>c_e(x_e)$.
Thus, by the monotonicity of $c_e$ this implies $\bar x_e>x_e$. By elementary transformations of bases in polymatroid base polytopes 
(see Murota~\cite[Theorem 4.3]{murota2003discrete}), there must
exist $f\in E\setminus\{e\}$ with $\bar x_f<x_f$ and $\epsilon>0$ such that
\[ x+\epsilon(\chi_e-\chi_f)\in P_\rho \text{\ \  and\ \ }
\bar x+\epsilon(\chi_f-\chi_e)\in P_\rho.\]
As $x$ and $\bar x$ are both optimal solutions for their
respective optimization problems,
and $x_f, \bar{x}_e >0$, we obtain by the optimality conditions~\eqref{thm:fuj} that
$c_f(x_f) \leq c_e(x_e)$ and
$\bar{c}_e(\bar{x}_e)\leq \bar{c}_f(\bar{x}_f)$.
Hence,
\[ c_f(x_f) 
\leq   c_e(x_e)< \bar c_e(\bar x_e)
\leq\bar c_f(\bar x_f)\leq c_f(x_f),\]
a contradiction.
\end{proof}

We now prove a second sensitivity result for minimizers of the Beckman function
over the polymatroid base polytope for the case 
when the demand is reduced.
 \begin{lemma}\label{lem:dChangePolymat}
Let $I$ and $\bar{I}$ be two instances of problem~\eqref{potential},
with the only difference that for $\bar{I}$ the demand for one
population  
$j\in N$ is decreased from $d_j$ to $\bar{d_j}< d_j$.
Hence, the feasible strategy distributions
$P$ and $\bar{P}$ of
$I$ and $\bar{I}$, respectively,
are given by
\begin{align*}
P &= P_\rho = \sum_{i\in N} P_i, \text{ and}\\
\bar{P} &= \frac{\bar{d_j}}{d_j}P_j +
 \sum_{i\in N\setminus \{j\}} P_i.
\end{align*}
Let $x, \bar{x}\in \mathbb{R}^E_+$ be
minimizers of problem~\eqref{potential} over the
polytopes $P$ and $\bar{P}$, respectively. Then
\begin{equation*}
c_e(\bar{x}_e) \leq c_e(x_{e}) \qquad \forall e\in E.
\end{equation*}
\end{lemma}

\begin{proof}
By contradiction assume there is $e\in E$ with $c_e(\bar{x}_{e}) > c_e(x_{e})$. Since $c_e$ is
nondecreasing, this implies $\bar{x}_{e} > x_{e}$.
Because $x\in P=\sum_{i\in N}P_i$, we can decompose
$x$ as $x=\sum_{i\in N} x_i$ where $x_i \in P_i$
for $i\in N$.
Let $x'=\sum_{i\in N\setminus \{j\}} x_i +\frac{\bar{d}_j}{d_j}x_j\in \bar{P}$.
Clearly, $x'\leq x$ component-wise, and we thus have
in particular $\bar{x}_{e} > x_{e} \geq x'_{e}$.
By exchange properties of base polytopes of polymatroids
there exist 
$f\in E$ with
\begin{equation}\label{eq:xbLeqXp}
\bar{x}_{f} < x'_{f} 
\quad\text{and}
\end{equation}
$\epsilon >0$ such that
\begin{enumerate}[(i)]
\setlength\itemsep{0em}
\item\label{item:xbChange}
 $\bar{x} + \epsilon (\chi_f - \chi_e) \in \bar{P}$, and
\item\label{item:xpChange}
 $x'      + \epsilon (\chi_e - \chi_f) \in \bar{P}$.
\end{enumerate}
Notice that $\bar{P} + (1-\frac{\bar{d_j}}{d_j})P_j = P$,
and hence, $x-x' \in (1-\frac{\bar{d_j}}{d_j})P_j$.
Together with~\eqref{item:xpChange} this implies
\begin{equation*}
x+\epsilon(\chi_e - \chi_f) =
\underbrace{x'+\epsilon(\chi_e - \chi_f)}_{\in \bar{P}}
 + \underbrace{(x-x')}_{\in \left(1-\frac{\bar{d_j}}{d_j}\right)P_j}\in P,
\end{equation*}
and therefore
\begin{equation}\label{eq:xEqu}
c_e(x_{e}) \geq c_f(x_{f}),
\end{equation}
since $x$ is a minimizer of $\Psi$ over $P$.
Similarly,~\eqref{item:xbChange} implies
\begin{equation}\label{eq:xbEqu}
c_e(\bar{x}_{e}) \leq c_f(\bar{x}_{f}),
\end{equation}
because $\bar{x}$ is a minimizer of $\Psi$
over $\bar{P}$.
Putting things together, we obtain
\begin{align*}
c_e(x_{e}) &< c_e(\bar{x}_{e}) && \text{(by assumption for sake
  of contradiction)} \\
   &\leq c_f(\bar{x}_{f})
       && \text{\eqref{eq:xbEqu}} \\
   &\leq c_f(x'_{f})
       && \text{(\eqref{eq:xbLeqXp} and $c_f$ is nondecreasing)}\\
   &\leq c_f(x_{f})
       && \text{($x'\leq x$ and $c_f$ is nondecreasing)}\\
   &\leq c_e(x_{e})
       && \text{\eqref{eq:xEqu}},
\end{align*}
thus leading to a contradiction and finishing the proof.
\end{proof}
Furthermore, 
Lemma~\ref{lem:cost-sen} shows that equilibrium costs only decrease (component-wise) when reducing the cost function.
Hence, reducing costs and demands simultaneously
can be interpreted as doing
first one reduction and then the other. Each of these changes
is such that any new Wardrop equilibrium has on any resource a lower
cost than before the change. Thus, 
a family of 
matroid set systems (as a special case of polymatroids) is 
immune to the weak Braess paradox.

\paragraph{From Clutters to Set Systems.}
To complete the direction (I) $\Rightarrow$ (II)
of the proof of Theorem~\ref{thm:braess}, it remains
to show that if the clutters $(E,(\S_i)^{\min})_{i\in N}$ of some
set systems $(E,\S_i)_{i\in N}$ are immune to the weak Braess paradox
then so are the original set systems $(E,\S_i)_{i\in N}$.
%

For any set $U\in \S_i$, we call $U^*\in S_i$ a
\emph{tight subset} of $U$, if $U^*$ is an inclusion-wise
minimal set in $\S_i$ contained in $U$, i.e,
$U^*\subseteq U$ and $U^* \in (\S_i)^{\min}$.

\begin{lemma}\label{lem:clutter}
If $(E,(\S_i)^{\min})_{i\in N}$ is immune to the weak Braess paradox
then so is $(E,\S_i)_{i\in N}$.
\end{lemma}
\begin{proof}
Let $\M=(N,E,(\S_i)_{i\in N}, (c_e)_{e\in E}, (d_i)_{i\in N})$
and
$\bar{\M}=(N,E,(\S_i)_{i\in N}, (\bar{c}_e)_{e\in E}, (\bar{d}_i)_{i\in N})$
be two congestion models defined on the same strategy spaces 
$(\S_i)_{i\in N}$ and satisfying $\bar{c}_e(t) \leq c_e(t)$ for
$e\in E, t\geq 0$ and $\bar{d}_i\leq d_i$ for $i\in N$.
Moreover, let 
$\M^{\min}$ 
and $\bar \M^{\min}$  
be the models obtained from $\M$ and $\bar{\M}$, respectively,
by replacing $(\S_i)_{i\in N}$ by $((\S_i)^{\min})_{i\in N}$.

Let $x\in P(\M)$ and $\bar x\in P(\bar \M)$ be two Wardrop equilibria.
We have to show $c_e(x_e) \geq \bar{c}_e(\bar{x}_e)$ for $e\in E$.
Starting from $x$, 
we iteratively pick $i\in N, U\in \S_i\setminus {(\S_i)^{\min}}$ with $x^i_{U}>0$ and  a tight subset $U^*$ of $U$, and 
set 
\begin{itemize}
\item $x^i_{U^*}\leftarrow x^i_{U^*}+x^i_{U},\quad x^i_U\leftarrow 0$.
\end{itemize}
We denote by $x'$ the thus obtained point, whose load vector
satisfies $(x'_e)_{e\in E}\in \tilde P(\M^{\min})\subseteq \tilde P(\M)$.
We apply the same procedure to $\bar x$ to obtain ${\bar x}'$ 
with $({\bar x}'_e)_{e\in E}\in \tilde P(\bar \M^{\min}))\subseteq \tilde P(\bar \M)$.
The condition $U^*\subseteq U$ implies that the profiles  $x'$ and $\bar x'$ satisfy
\[ c_e(x'_e)\leq c_e(x_e) \text{ and } \bar c_e(\bar x'_e)\leq \bar c_e(\bar x_e) \text{ for all }e\in E.\]
This implies that $x'$ and $\bar x'$ are global minimizers of the Beckmann potential
over both polytopes $\tilde P(\M)$ and $\tilde P(\M^{\min})$, and, $\tilde P(\bar\M)$ and $\tilde P(\bar\M^{\min})$, respectively.
 Thus, $x'$ and $\bar x'$ are  
equilibrium load vectors for both models $\M$ and $\M^{\min}$, and, $\bar \M$ and $\bar\M^{\min}$, 
respectively.
Since equilibrium costs per resource are unique (see Remark~\ref{l.potential}), we get
\[ c_e(x'_e)= c_e(x_e) \text{ and } \bar c_e(\bar x'_e)= \bar c_e(\bar x_e) \text{ for all }e\in E.\]
As by assumption $(E,(\S_i)^{\min})_{i\in N}$ is immune to
the weak Braess paradox we have $c_e(x'_e) \geq \bar{c}_e(\bar{x}_e')$
for $e\in E$, which finally implies
\[ c_e(x_e)=c_e(x'_e)\geq  \bar c_e(\bar x'_e)=\bar c_e(\bar x_e) \text{ for all }e\in E,\]
proving the lemma.
\end{proof}

\subsection{Proof of Theorem~\ref{thm:braess}: (II) $\Rightarrow$ (I).}
Let us recall the definition of a set system being immune
to the weak Braess paradox.
Let $\M$ and $\bar \M$ be any two models of a nonatomic
congestion game with the only difference that
for $\bar \M$ we use cost functions satisfying $\bar c_e(t)\leq c_e(t)$ 
for all $e\in E$ and $t\geq 0$ and $\bar d_i\leq d_i$ for all $i\in N$.
Then, if the strategy spaces of $\M, \bar\M$ are immune to the
weak Braess paradox, we have that 
any two Wardrop equilibria $x$ and $\bar x$ for $\M$
and $\bar \M$, respectively, 
satisfy
\begin{align}\label{wBP-free} \bar c_e(\bar x_e)\leq c_e(x_e), \text{ for all }e\in E.\end{align}

We first give an outline of the proof. Let $(\S_i)_{i\in N}$
be strategy spaces that are immune to the weak Braess paradox.
We fix the demand of each population to $1$, i.e., $d_i=1$ for
$i\in N$. Hence, the possible load vectors that can be obtained
by population $i\in N$ by playing a feasible strategy are given
by the following polytope
\begin{equation*}
P_i = \text{Convex hull of }\{\chi_S \mid S\in \S_i\}
\qquad \forall i\in N,
\end{equation*}
where $\chi_S$ denotes the characteristic vector of $S\subseteq E$.
Thus, the load vectors obtainable by considering all populations
combined are described by the Minkowski sums of the polytopes
$P_i$, i.e.,
\begin{equation*}
\tilde{P}(\M) = \sum_{i\in N} P_i.
\end{equation*}
For brevity, we set
\begin{equation*}
P = \tilde{P}(\M).
\end{equation*}
Analogously, we define the polytopes corresponding to the possible
load vectors stemming from the strategy spaces
$((\S_i)^{\min})_{i\in N}$, i.e.,
\begin{equation*}
P_i^{\min} = \text{Convex hull of }\{\chi_S \mid S\in (\S_i)^{\min}\}
\qquad \forall i \in N, 
\end{equation*}
and we let
\begin{equation*}
P_i^{\uparrow} = P_i + \mathbb{R}_{\geq 0}^E = P_i^{\min} + \mathbb{R}_{\geq 0}^E \qquad \forall i\in N
\end{equation*}
be the dominant of $P_i$ for $i\in N$.
Observe that
\begin{equation}\label{eq:domPRelDomPi}
P^{\uparrow} = \sum_{i\in N} P_i^{\uparrow}.
\end{equation}

We will use a polyhedral approach to show (II) $\Rightarrow$ (I).
More precisely, we will show that $(\S_i)_{i\in N}$ being immune
to the weak Braess paradox implies that each $P^{\min}_i$ for $i\in N$
is the base polytope of a matroid.
To show that $P^{\min}_i$ for $i\in N$ is the base polytope
of a matroid we rely on the following well-known characterization
(see, e.g., \cite[Theorem 17.1]{fujishige2005submodular}), and a corollary thereof.

\begin{lemma}\label{lem:charMatBasePolytope}
Let $Q\subseteq [0,1]^E$ be a $\{0,1\}$-polytope, i.e., all
vertices of $Q$ are part of $\{0,1\}^E$. Then $Q$ is the base
polytope of a matroid if and only if each edge direction
of $Q$ is parallel to some vector $\chi_{e} -\chi_{f} $ for distinct
$e,f\in E$.
\end{lemma}

The above lemma can be rephrased in terms of the dominant $Q^{\uparrow}$ of $Q$ as follows.
\begin{corollary}\label{cor:charDomMatBasePolytope}
Let $Q\subseteq [0,1]^E$ be a $\{0,1\}$-polytope such that no two distinct vertices $u,v$ of $Q$ satisfy $u\leq v$ (component-wise). Then $Q$ is the base polytope of a matroid if and only if each edge direction of $Q^{\uparrow}$ is parallel to some vector $\chi_e - \chi_f$ for distinct $e,f\in E$.
\end{corollary}
\begin{proof}
If $Q$ is a matroid polytope then, by Lemma~\ref{lem:charMatBasePolytope}, each edge direction of $Q$ is parallel to $\chi_e - \chi_f$ for distinct $e,f\in E$. Since each edge direction of the dominant of a polytope is also an edge direction of the original polytope, we have that the edge directions of $Q^{\uparrow}$ are also of this form.

Conversely, assume that each edge direction of $Q^{\uparrow}$ is parallel to $\chi_e - \chi_f$ for some distinct $e,f\in E$. 
First observe that $Q$ and $Q^{\uparrow}$ have the same set of vertices. This follows by the fact that no two distinct vertices $u,v$ of $Q$ satisfy $u\leq v$ (component-wise).
Moreover, each edge direction of $Q^{\uparrow}$ is parallel to some vector of the form $\chi_{e} - \chi_{f} $ for distinct $e,f\in E$.
Since the set of all the vertices of $Q^{\uparrow}$ is linked (connected) by the edges of $Q^{\uparrow}$ (see, e.g.,~\cite{balinski_1961_graph}),  this implies that all vertices of $Q^{\uparrow}$ have the same $\ell_1$-norm and thus lie on a hyperplane of the form $\langle -\bm{1}, x\rangle = b$, where $\bm{1}\in\mathbb{R}^E$ is the all-ones vector. Hence, $\langle -\bm{1}, x \rangle = b$ is a supporting hyperplane of $Q^{\uparrow}$ and the face of $Q^{\uparrow}$ defined by this hyperplane is equal to $Q$.
Since $Q$ is a face of $Q^{\uparrow}$, each edge of $Q$ is also an edge of $Q^{\uparrow}$, and is thus parallel to $\chi_e - \chi_f$ for distinct $e,f\in E$. Lemma~\ref{lem:charMatBasePolytope} now implies that $Q$ is indeed a matroid base polytope.
\end{proof}

Hence, it suffices to show that any edge direction of $P_i^{\uparrow} = (P_i^{\min})^{\uparrow}$ for $i\in N$ is of type $\chi_e - \chi_f$ for distinct $e,f\in E$. Notice that Corollary~\ref{cor:charDomMatBasePolytope} can indeed be applied to the polytopes $(P_i^{\min})^{\uparrow}$, because any two distinct vertices $u,v$ of $P_i^{\min}$ correspond to characteristic vectors of two distinct sets in the clutter $(E,(\mathcal{S}_i)^{\min})$, which implies $u\not\leq v$ as required by Corollary~\ref{cor:charDomMatBasePolytope}.

The following well-known property links edge directions of a Minkowski sum of polytopes to the edge directions of the summands (see, e.g., \cite[Lemma~6.14]{thomas_1998_applications}).
\begin{lemma}\label{lem:minkowskiED}
Let $Q_1,\dots, Q_k \subseteq \mathbb{R}^E$ be polytopes
and $Q=\sum_{i=1}^k Q_i$ be their Minkowski sum. Then the
set of edge directions of $Q$ is the union of the sets
of edge directions of the $Q_i$ for $i\in [k]$.
\end{lemma}

However, we are interested in the edge directions of $P_i^{\uparrow}$ for $i\in N$, which are unbounded polyhedra. Lemma~\ref{lem:minkowskiED} can easily be generalized to dominants of polytopes as shown below. This allows us to derive properties on the edge directions of $P^{\uparrow}_i$ for $i\in N$ from properties of edge directions of their Minkowski sum $P^{\uparrow}$.

\begin{lemma}\label{lem:minkowskiEDDom}
Let $Q_1,\ldots, Q_k \subseteq \mathbb{R}^E$ be polytopes and $Q=\sum_{i=1}^k Q_i$ be their Minkowski sum. As usual, let $Q_1^{\uparrow}, \ldots, Q_k^{\uparrow}$ and $Q^{\uparrow}$ denote the dominants of $Q_1, \ldots, Q_k$ and $Q$. Then the set of edge directions of $Q^{\uparrow}$ is the union of the sets of edge directions of the $Q_i^{\uparrow}$ for $i\in [k]$.
\end{lemma}
\begin{proof}
Let $j\in [k]$ and consider an edge $\overline{uv}$ of $Q^{\uparrow}_j$. We first show that $Q^{\uparrow}$ contains an edge parallel to $\overline{uv}$. 
There is a non-zero vector $h\in \mathbb{R}^E$ such that $\overline{uv}$ are the set of all maximizers of $\langle h,x \rangle$ for $x\in Q^{\uparrow}_j$. Because $Q_j^{\uparrow}$ is a dominant, we must have $h < 0$. For $i\in [k]$, let $Q_i^h\subseteq Q^{\uparrow}_i$ be the set of all maximizers in $Q^{\uparrow}_i$ of the linear function $\langle h, x \rangle$. Analogously, let $Q^h\subseteq Q^{\uparrow}$ be all points in $Q^{\uparrow}$ maximizing $\langle h, x\rangle$. Because the maximizers of a Minkowski sum are the sum of the maximizers of the summands, we have $Q^h = \sum_{i=1}^k Q_i^h$. Moreover, because $Q_1^{\uparrow}, \ldots, Q_k^{\uparrow}$, $Q^{\uparrow}$ are all dominants of polytopes and $h<0$, we have that $Q_1^h,\ldots, Q_k^h$ and $Q^h$ are also polytopes. Applying Lemma~\ref{lem:minkowskiED} to $Q^h = \sum_{i=1}^k Q_i^h$, we obtain that $Q^h$ contains an edge parallel to $\overline{uv}$. Finally, since $Q^h$ is a face of $Q^{\uparrow}$, every edge of $Q^h$ is also an edge of $Q^{\uparrow}$, implying that $Q^{\uparrow}$ contains an edge parallel to $\overline{uv}$, as desired.

Conversely, we have to show that for any edge $\overline{uv}$ of $Q^{\uparrow}$ there is an index $j\in [k]$ such that $Q^{\uparrow}_j$ contains an edge parallel to $\overline{uv}$. This can be proven analogously to the previous case, by considering a vector $h<0$ such that $\overline{uv}$ are all maximizers of $\langle h,x \rangle$ for $x \in Q^{\uparrow}$, and considering for each $Q_1^{\uparrow},\ldots,Q_k^{\uparrow}$ and $Q^{\uparrow}$ the face corresponding to all maximizers of $\langle h, x\rangle$. The result then again follows by Lemma~\ref{lem:minkowskiED}.
\end{proof}

Hence, Corollary~\ref{cor:charDomMatBasePolytope} and Lemma~\ref{lem:minkowskiEDDom} imply that to prove the direction (II) $\Rightarrow$ (I) of Theorem~\ref{thm:braess}, it suffices to show the following.

\begin{proposition}\label{prop:EDofPup}
Each edge direction of $P^{\uparrow}$ is parallel
to some vector $\chi_{e}  - \chi_{f} $ with distinct $e,f\in E$.
\end{proposition}

Indeed, proving Proposition~\ref{prop:EDofPup} implies
by Lemma~\ref{lem:minkowskiEDDom} that each edge direction
of $P^{\uparrow}_i$ for $i\in N$ is parallel to
$\chi_{e} -\chi_{f} $ for some distinct $e,f\in E$,
which, by applying Corollary~\ref{cor:charDomMatBasePolytope} to $Q=P^{\min}_i$ (notice that $(P^{\min}_i)^{\uparrow} = P_i^{\uparrow}$),
implies that $P^{\min}_i$
for $i\in N$ are matroid base polytopes as desired.
Hence, it remains to prove Proposition~\ref{prop:EDofPup}.


We start with a simple observation.

\begin{lemma}\label{lem:EDofPUp}
Let $\overline{uv}$ be an edge of $P^{\uparrow}$, then
$\overline{uv}$ is parallel to some vector
in $\{-1,0,1\}^E$ with at least $2$ non-zero values.
\end{lemma}
\begin{proof}
Notice that every edge of $P^{\uparrow}$ must also be an edge of $P$.
Hence, by Lemma~\ref{lem:minkowskiED} the edge $\overline{uv}$ (of $P$)
must be parallel to some edge direction of one of the
polytopes $P_i$ for some $i\in N$. The statement now follows
by observing that all edge directions of $P_i$ are within
$\{-1,0,1\}^E$ because $P_i$ is a $\{0,1\}$-polytope;
moreover, the edge direction cannot be of the form
$\chi_{e} $ since $P^{\uparrow}$ is up-closed,
implying that any edge direction of $P^{\uparrow}$ must
have at least $2$ non-zero values.
\end{proof}

The following lemma shows the existence of particular
types of cost functions which will help us to derive
further properties on edge directions by using the fact
that no weak Braess paradox exists for the strategy
spaces $(\S_i)_{i\in N}$.

\begin{lemma}\label{lem:costForPUp}
Let $x^0$ be the midpoint of an edge $\overline{uv}$ of
$P^{\uparrow}$, and let $w\in \{-1,0,1\}^E$ be a vector
parallel to $\overline{uv}$.
Then there exist strictly positive, continuous and strictly
increasing
cost functions $c_e$ for $e\in E$ such that
the following holds:
\begin{enumerate}
\item\label{item:x0UniqueMin}
The Beckmann potential $\Phi(x)$ that corresponds to
$(c_e)_{e\in E}$ has $x^0$ as a unique minimizer
over $P^{\uparrow}$.

\item\label{item:uniqueMinMoves}
For any $f\in E$ with $w_f=1$ and
$\alpha >0$ let $c_e^{(\alpha,f)}$ for $e\in E$
be the following cost function:
\begin{equation*}
c_e^{(\alpha,f)}(t) = \begin{cases}
c_e(t) & \text{if } e\in E\setminus \{f\},\\
c_e(t) - \alpha & \text{if } e=f.
\end{cases}
\end{equation*}
Then for a sufficiently small $\alpha >0$,
the costs $(c_e^{(\alpha,f)})_{e\in E}$ are
non-negative, continuous and strictly increasing and
the corresponding Beckmann potential
$\Phi^{(\alpha,f)}(x)$ has a unique minimizer
over $P^{\uparrow}$ which is of the form
$x^0 + \beta w$ for some $\beta >0$.
\end{enumerate}
\end{lemma}
\begin{proof}
Let $\langle h, x\rangle =b$ be a supporting hyperplane
of $P^{\uparrow}$ defining the edge $\overline{uv}$. Hence, $\overline{uv}$ are the set of all maximizers of $\langle h, x\rangle$ for $x\in P^{\uparrow}$.
Notice that $h<0$ (component-wise) because $P^{\uparrow}$ is a dominant. Indeed $h_e >0$ for some $e\in E$ is not possible because in this case $\overline{uv}$ cannot be maximizers with respect to $h$ as increasing the component corresponding to $e$ leads to better points. Moreover, there is also no $e\in E$ with $h_e=0$, because this would again imply that, given any point $y$ on $\overline{uv}$, one can increase the $e$-component of $y$ arbitrarily and remain a maximizer with respect to $h$. However, as $\overline{uv}$ is an edge, it is bounded and it is therefore not possible to increasing any component arbitrarily and remain on the set of maximizers $\overline{uv}$. Thus, $h<0$.

Furthermore, we define
\[\gamma = \min\left\{\frac{-h_e}{4 x^0_e} \middle \vert e\in E,x^0_e>0 \right\},\]
and for $t \geq 0$, we let
\begin{equation*}
c_e(t) = -h_e + 2\gamma (t-x^0_e),
\end{equation*}
which leads to the Beckmann potential
\begin{equation*}
\Phi(x) = -\langle h , x \rangle + \gamma \|x-x^0\|^2_2-\gamma\|x^0\|^2_2.
\end{equation*}
Clearly, the cost functions $c_e$ are strictly positive,
continuous and strictly increasing  (recall that $h_e<0$ for all $e\in E$).
Moreover, $x^0$ is indeed the unique minimizer
of $\Phi(x)$ on $P^{\uparrow}$ since it is the unique
minimizer of $\gamma \|x-x^0\|_2^2$ over $\mathbb{R}^E$
and $x^0$ is a maximizer of $\langle h, x\rangle$
over $P^{\uparrow}$.
This shows~\ref{item:x0UniqueMin}.

To show~\ref{item:uniqueMinMoves} we show that for a sufficiently
small $\alpha >0$ the Beckmann potential $\Phi^{(\alpha,f)}(x)$
has as unique minimizer the point
$y = x^0 + \beta_\alpha w$, where
\begin{equation*}
\beta_\alpha = \frac{\alpha}{2 \gamma \|w\|_2^2}.
\end{equation*}
We will choose $\alpha >0$ small enough such that
$y$ is still in the relative interior of the edge
$\overline{uv}$.

Notice that by definition of the Beckmann potential
its gradient at a point $x$ is given
by the cost functions, hence,
\begin{equation*}
\nabla \Phi^{(\alpha,f)} (x) = 
(c^{(\alpha,f)}_e(x))_{e\in E}.
\end{equation*}

Since $\Phi^{(\alpha,f)}(x)$ (like $\Phi(x)$) is a strictly
convex function, it has a unique minimizer over $P^{\uparrow}$.
Moreover, by the Karush-Kuhn-Tucker conditions we have that
$y$ is a minimizer of $\Phi^{(\alpha,f)}(x)$ over $P^{\uparrow}$
if and only if $-\nabla \Phi^{(\alpha,f)}$ is spanned by
the cone $\mathcal{C}(y)\subseteq \mathbb{R}^E$
of the normal vectors of the linear constraints
of $P^{\uparrow}$ that are tight at $y$. To be precise,
for the above statement to hold, we consider
an inequality description of $P^{\uparrow}$ where, in particular,
a linear equality constraint is represented by two opposing
inequality constraints.
Notice that $\mathcal{C}(x)$ is the same cone for
any point $x$ in the relative interior of the edge $\overline{uv}$,
since all points in the relative interior of $\overline{uv}$
have the same set of tight constraints of $P^{\uparrow}$.

Furthermore, $h$ is in the relative interior of
$\mathcal{C}(y)$ because the supporting hyperplane
$\langle h,x \rangle = b$ defines the edge $\overline{uv}$.
This implies that there exists an $\epsilon >0$ such that
any vector $h'\in \mathbb{R}^E$ with $\|h-h'\|_2 \leq \epsilon$
and $h' \perp w$ satisfies $h'\in \mathcal{C}(y)$.
We show that for a small enough $\alpha$,
the negative gradient $-\nabla \Phi^{(\alpha,f)} (y)$ can be chosen
as such an $h'$.

We indeed have $-\nabla\Phi^{(\alpha,f)}(y) \perp w$ since
\begin{align*}
\langle \nabla \Phi^{(\alpha,f)} (y), w\rangle &=
  -\langle h,w\rangle + 2 \gamma \beta_\alpha \|w\|_2^2 - \alpha \\
   &= 0.  && \text{(using $h \perp w$ and definition of $\beta_\alpha$)}
\end{align*}
Furthermore since $\nabla \Phi^{(\alpha, f)}(y)$ is continuous
in $\alpha$, and $-\nabla \Phi^{(0,f)}(y) = -\nabla \Phi (x^0) = h$,
we have for a small enough $\alpha > 0$ that
$\| h-(-\nabla \Phi^{(\alpha, f)} (y))\|_2 \leq \epsilon$,
as desired.
\end{proof}

We finally complete the proof of (II) $\Rightarrow$ (I) by
showing Proposition~\ref{prop:EDofPup}.
\begin{proof}(of Proposition~\ref{prop:EDofPup})
Assume by the sake of contradiction that the proposition
does not hold.
Hence, there is an edge $\overline{uv}$ of $P^{\uparrow}$,
which is parallel to some vector $w\in \{-1,0,1\}^E$
with at least $2$ non-zero entries by Lemma~\ref{lem:EDofPUp},
and such that $w$ has either two $1$-entries or two
$-1$-entries. Without loss of generality we can assume
that $w$ has two $1$-entries (for otherwise consider $-w$).
Let $f,g\in E$ be two such entries, i.e., $w_f = w_g = 1$.
We now invoke Lemma~\ref{lem:costForPUp} with respect
to the edge $\overline{uv}$ and the element $f\in E$
to obtain the following.
There are cost functions $(c_e)_{e\in E}$ such
that the corresponding Beckmann potential $\Phi(x)$
has its unique minimum at the midpoint of
the edge $\overline{uv}$, denoted by $x^0$.
Moreover, it suffices to reduce the cost
function for element $f$ to obtain new cost
functions $(\bar{c}_e)_{e\in E}$ with a new corresponding
potential
($\Phi^{(\alpha,f)}$ in Lemma~\ref{lem:costForPUp})
whose unique minimum over
$P^{\uparrow}$ is equal to $y=x^0 + \beta w$
for some $\beta >0$.
However, since $w_g=1$, this implies that
$y_g > x^0_g$, and hence
\begin{align*}
\bar{c}_g(y_g) = c_g(y_g) > c_g(x^0_g),
\end{align*}
where the second equality follows from the fact that
$c_g$ is strictly increasing by Lemma~\ref{lem:costForPUp}.
Hence, this shows that we have the weak Braess paradox
if the combined strategy spaces are described by
$P^{\uparrow}$.

However, because the unique minimizers
of the two considered Beckmann potentials above
over $P^{\uparrow}$ are both part of $P$---
since they lie on $\overline{uv}$, which is an edge
of $P$---they are also the unique minimizers of the
same Beckmann potentials over the smaller set $P$.
This implies that we have the weak Braess paradox
over the strategy spaces $(\S_i)_{i\in N}$, which is
a contradiction.
\end{proof}

\section{The Strong Braess Paradox for Non-Matroid Set Systems.} \label{sec:strongBP}
We now investigate a combinatorial property of the set systems so that the strong Braess
paradox does not occur. 
In contrast to
the weak Braess paradox, the matroid property is not necessary
for immunity against the strong Braess paradox.
Milchtaich~\cite{Milchtaich06GEB}, for example, shows that if the strategy space of every player is symmetric
and corresponds to the paths of a series-parallel $s$-$t$ graph, then
there will be no strong Braess paradox. Note that in this case the
resulting set systems need not be bases of matroids.

In this section, we derive a characterization of the 
strong Braess paradox that does not take into account the global structure of the game.
Specifically,  we show that the matroid property is the maximal condition on the players' strategy spaces that guarantees that the strong Braess paradox does not occur \emph{without} taking into account how the strategy spaces of different players interweave (cf.~Ackermann, R\"oglin and V\"ocking~\cite{Ackermann08,Ackermann09} who introduced the notion of interweaving of strategy spaces).
To state this property mathematically precisely, 
we introduce the notion of \emph{embeddings}.
Let $\tilde{E}=\{e_1, \ldots, e_p\}$ be a set consisting of $p=\sum_{i\in N}|E_i|$ elements,
where, as before, $E_i=\cup_{S\in\S_i}, i\in N$. 
Formally, an embedding is a map $\tau:=(\tau_i)_{i\in N}$, where every $\tau_i: E_i \rightarrow \tilde{E}$ is an injective map from $E_i$ to $\tilde{E}$.
The embedding of $\S_i$ in $\tilde{E}$ according to $\tau$ is then defined by
identifying every $S=\{e_1,\dots, e_k\}\in \S_i$ with $\tau_i(S):=\{\tau_i(e_1),\dots, \tau_i(e_k)\}$
and $\tau(\S_i):=\{\tau_i(S)|S\in \S_i\}$.
Given $(\S_i)_{i\in N}$ and $\tau$, the new combined strategy space is then denoted by $(\tau_i(\S_i))_{i\in N}$. 

\begin{definition}\label{def:embeddings}
A family of set systems $(E,\S_i)_{i\in N}$ with $\S_i\subseteq 2^{E}, i\in N$ is said to be 
\emph{universally} immune to the strong Braess paradox if 
for all embeddings $\tau$ in $\tilde{E}$, the set systems $(\tilde{E},\tau_i(\S_i))_{i\in N}$
do not admit the strong Braess paradox {\rm (}in the sense of 
Definition~{\rm \ref{def:braess})}.
\end{definition}

Since any embedding of the set of bases of a matroid into a ground set
of resource is a set of bases of a matroid again, we obtain the
following immediate consequence of Theorem \ref{thm:braess}.
\begin{corollary}\label{cor:strongBP}
If for each $i\in N$ the clutter $(E,(\S_i)^{\min})$ forms the base set 
of a matroid $M_i=(E,\I_i)$, then the family of set systems $(E,\S_i)_{i\in N}$
is universally immune to the strong Braess paradox.
\end{corollary}

Our second result now gives a complete characterization of set systems
that are {\it universally} immune to the strong Braess paradox.
\begin{framed}
\begin{theorem}\label{thm:braess-strong}
Let $|N|\geq 2$ and $(E,\S_i)_{i\in N}$ with $\S_i\subseteq 2^{E}\setminus{\{\emptyset\}}$ for each  $i\in N$.
Then, the following three statements are equivalent.
\begin{enumerate}
\item[{\rm (I)}]
$(E,(\S_i)^{\min})$ forms the base set 
of a matroid $M_i=(E,\I_i)$ for each $i\in N$.
\item[{\rm (II)}]
$(E,\S_i)_{i\in N}$ is immune to the weak Braess paradox.
\item[{\rm (III)}]
$(E,\S_i)_{i\in N}$ is universally immune to the strong Braess paradox.
\end{enumerate}
\end{theorem}
\end{framed}

(I) $\Leftrightarrow$ (II): See Theorem~\ref{thm:braess}.

(I) $\Rightarrow$ (III):  See Corollary~\ref{cor:strongBP}.

We prove (III) $\Rightarrow$ (I) by contradiction.
Consider a family of non-empty set systems $(E,\S_i)_{i\in N}$ with $n:=|N|\ge 2$, and assume that at least one of the induced clutters $(E, (\S_i)^{\min})_{i\in N}$, 
say $(E, (\S_1)^{\min})$, 
is \emph{not} the base set of a matroid.
As above, let $E_i:=\bigcup_{S\in \S_i} S$ denote the set of those resources that occur in at least one set in $\S_i$.
We will show that the family of set systems $(E,\S_i)_{i\in N}$ admits embeddings $\tau_i:E_i\to \tilde{E}$, $i\in N$, such that
$\tau(\S)=(\tau_1(\S_1),\ldots, \tau_n(\S_n))$ admits the strong Braess paradox.

Let us call, in general, a non-empty clutter $(E, \mathcal{F})$ a \emph{non-matroid} if 
the  set system $(E, \{X\subseteq S : S\in \mathcal{F}\})$ is not a
matroid.

Our proof relies on a certain property of non-matroids stated in the following lemma.
Its proof
can also be derived from the proof of Lemma 5.1 in \cite{HarksP14}, or the proof of Lemma 16 in \cite{Ackermann09}.
\begin{lemma}\label{l.anti-matroid}
 If clutter $(E, \mathcal{F})$  with $\mathcal{F}\neq \emptyset$ is a non-matroid, then there exist $X,Y\in \mathcal{F}$ and $\{a,b,c\}\subseteq X\Delta Y:= (X\setminus Y) \cup (Y\setminus X)$
such that for each set $Z\in \mathcal{F}$ with $Z\subseteq X\cup Y$, either
$a\in Z$ or $\{b,c\} \subseteq Z$.
\end{lemma}
\begin{proof} 
Recall the \emph{basis exchange property} for matroids: a clutter $(E,\mathcal{F})$ is the family of bases of some matroid
if and only if for any $X, Y\in \mathcal{F}$ and $e\in X\setminus{Y}$ there exists some $f\in Y\setminus{X}$ such that $X-e+f\in \mathcal{F}$.
Thus, if the clutter $(E,\mathcal{F})$ is a non-matroid, there must exist $X,Y\in \mathcal{F}$ and $e\in X\setminus{Y}$
such that \emph{for all} $f\in Y\setminus{X}$ the set $X-e+f$ does \emph{not} belong to $\mathcal{F}$.
We choose such $X, Y$ and $e\in X\setminus{Y}$ with $|Y\setminus{X}|$ minimal 
(among all $Y'\in \mathcal{F}$ with $X-e+f'\not\in \mathcal{F}$ for all $y'\in Y'\setminus{X}$).
Note that $|Y\setminus{X}|\ge 1$, since $\mathcal{F}$ is a clutter.
We distinguish the two cases  $|Y\setminus{X}|=1$ and  $|Y\setminus{X}|> 1$:
In case $|Y\setminus{X}|=1$, set $\{a\}=Y\setminus{X}$ and choose any two distinct elements $\{b,c\}\in X\setminus{Y}$.
Note that  $|X\setminus{Y}|\ge 2$ as otherwise, if $X\setminus{Y}=\{e\}$, then $Y=X-e+a$, in contradiction to our assumption.
Now, for any set $Z\subseteq (X\cup Y)-a$ with $Z\in \mathcal{F}$, the clutter property implies $Z=X$, and therefore $\{b,c\}\subseteq Z$, as desired.

In the latter case $|Y\setminus{X}|> 1$, we choose any two distinct elements $\{b,c\}\in Y\setminus{X}$ and set $a=e$.
Consider any $Z\in \mathcal{F}$ with $Z\subseteq (X\cup Y)-a$ and suppose, for the sake of contradiction, that $\{b,c\}\not\subseteq Z$.
Since $Z\setminus{X}\subseteq Y\setminus{X}$, there cannot exist some $g\in Z\setminus{X}$ with $X-a+g\in \mathcal{F}$.
However, $|Z\setminus{X}| < |Y\setminus{X}|$ in contradiction to our choice of $Y$.
\end{proof}

Using this property of non-matroids, we now define  embeddings
$\tau(\S)=(\tau_1(\S_1),\ldots, \tau_n(\S_n))$ that admit the strong Braess paradox.
The rough idea can be described as follows: we choose the embeddings, demands, and cost-functions in such a way that
the first two populations are independent of the remaining populations, and such that the game of the first two populations is isomorphic to the
routing game illustrated in Figure \ref{fig:nonmatroid} which admits the strong Braess paradox.

\begin{figure}[h!]
\begin{center}
\begin{tikzpicture}
 
 \begin{scope}[yshift=4.2cm]
 
  \node (fr) at (0, -0.4) {};
   \draw (fr) +(0, -0.2) node[labeledNodeS] (s1) {$s_1$};
 \draw (s1) +(1.5, 0.75) node[labeledNodeS] (v1) {$s_2$}
 edge[normalEdgeF, <-] node[above] {$3$} (s1) ;
 \draw (s1) +(1.5, -0.75) node[labeledNodeS] (v2) {$t$}
 edge[normalEdgeF, <-] node[above] {$1$} (s1) 
 edge[normalEdgeF, <-] node[left] {$x$} (v1);
 ;
 \end{scope}
\end{tikzpicture}
\begin{tikzpicture}
 
 \begin{scope}[yshift=4.2cm]
 
  \node (fr) at (0, -0.4) {}; 
 \draw (fr) +(2, -0.2) node[labeledNodeS] (s1) {$s_1$};
 \draw (s1) +(1.5, 0.75) node[labeledNodeS] (v1) {$s_2$}
 edge[normalEdgeF, <-] node[above] {$0$} (s1) ;
 \draw (s1) +(1.5, -0.75) node[labeledNodeS] (v2) {$t$}
 edge[normalEdgeF, <-] node[above] {$1$} (s1) 
 edge[normalEdgeF, <-] node[left] {$x$} (v1);
 ;
 \end{scope}
\end{tikzpicture}
\end{center}
\caption{There are two populations $1$ and $2$ that want to send $1/2$ units of demand each  from $s_1$ and $s_2$, respectively,  to $t$.
In the left network there is a unique Wardrop equilibrium, where
each population uses their direct edge leading to a cost of $1$
for every agent of population $1$ and a cost of $1/2$
for the agents of population $2$. 
Decreasing the cost from $3$ to $0$
for the arc $(s_1,s_2)$ induces now the unique Wardrop equilibrium,
where agents of population $1$ now choose the path $(s_1,s_2,t)$.
As a consequence, the private cost of the players in population 2 increase from $\frac{1}{2}$ to $1$. Also the social cost increases from $\frac{3}{4}$ to $1$.
}
\label{fig:nonmatroid}
\end{figure}
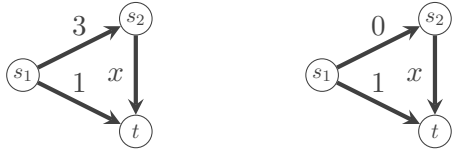

Let us set
the demands of all populations $d_i$ with $i\in N\setminus{\{1,2\}}$  to zero.
This way, the game is basically determined by 
the players in populations $1$ and $2$.
We set the demands $d_1=d_2=\frac{1}{2}$.

Let us choose two sets $X, Y$ in 
$(\S_1)^{\min}$ and
$\{a, b, c\}\subseteq X\cup Y$
as described in  Lemma~\ref{l.anti-matroid}.
Let $e:=\tau_1(a), f:=\tau_1(b)$ and $g:=\tau_1(c)$ with load-dependent costs $c_f(t)=t$, $c_g(t)=3$ and $c_e(t)=1$ for any $t\in \R_{\geq 0}$. We set the costs of all resources in $\tau_1(E_1)\setminus{(\tau_1(X)\cup \tau_1(Y))}$ to some very large cost $M$ (large enough so that no player of population $1$ would ever use any of these resources).
The cost on all resources in $(\tau_1(X)\cup \tau_1(Y))\setminus{\{e,f,g\}}$ is set to zero.
This way, each player of population $1$ always chooses a strategy $\tau_1(Z)\subseteq \tau_1(X)\cup \tau_1(Y)$ which, by Lemma~\ref{l.anti-matroid}, either contains
$e$, or it contains both $f$ and $g$.

In order to guarantee that each player of population $2$ always selects a strategy containing $f$, we select a set $S\in 
(\S_2)^{\min}$ of minimal cardinality,
and some arbitrary resource $k\in S$, 
and define the embedding $\tau_2$ such that $\tau_2(k)=f$ and $\tau_1(E_1)\cap \tau_2(E_2)=\{f\}$.
We set the resource costs such that
$c_r(x)=2$ for all resources $r\in \tau_2(E_2)\setminus{\{f\}}$.

Note that the game of population $1$ and $2$ can be represented by the routing game illustrated by the left network in Figure \ref{fig:nonmatroid}
if we interpret resource $e$ as arc $(s_1,t)$, resource $f$ as arc $(s_2,t)$, and resource $g$ as arc $(s_1,s_2)$.
By the choice of our cost functions, each player of population $2$ always selects the ``direct connection'', i.e., a strategy containing $f$, but neither $e$ nor $g$.
As long as the cost on $g:=(s_1,s_2)$ is $3$ (like in the left network),  each player of population $1$ selects the ``direct connection'', i.e.,  a strategy containing $e$, but neither $f$ nor $g$.
However, if the cost on $g$ is reduced from $3$ down to zero (like in the right network in Figure \ref{fig:nonmatroid}),
each player selects the seemingly cheaper strategy containing both, $f$ and $g$, but not $e$ (the ``indirect connection''), resulting in a Wardrop equilibrium in which each
player of population $2$ pays twice as much as in the Wardrop equilibrium for the left network, i.e., before the costs have been reduced. 
Not only the private cost of the players of population $2$, but also the total cost of the new Wardrop equilibrium increased after the cost on resource $g$ has been decreased.  Note that the entire construction
only involved cost reductions.

\begin{remark}
We can also characterize the universally strong Braess paradox
for the case where only demand reductions are considered. 
For this, we can use Lemma~{\rm \ref{l.anti-matroid}}
to construct a game that is isomorphic to the instance presented 
in Fig.~{\rm \ref{fig:braess-demand}}.
Note that in this instance we need two ``trivial'' players having one resource each,
thus, we need at least three populations. 
\end{remark}
\begin{remark}
Let us finally remark that the proof of  (III) $\Leftrightarrow$ (I) of Theorem~\ref{thm:braess-strong}
also works for characterizing set systems being universally immune
to a ``global'' version of the Braess paradox, where instead of population specific
costs, the total (social) cost is considered. Indeed, the counterexample derived
in the proof of (III) $\Rightarrow$ (I)  shows that for any non-matroid set system,
there exists an embedding admitting this global form of Braess paradox.
\end{remark}

\section{Conclusions}
Our results give a characterization of the weak Braess paradox for arbitrary set systems:
For any set system that is immune to the weak Braess paradox, the corresponding clutters must correspond to bases of some matroid. For the strong Braess paradox, we only 
derived a weaker characterization requiring the flexibility of arbitrary embeddings
of set systems into the ground set of resources. Characterizing the strong Braess paradox without
the use of embeddings remains an open question.

Our characterization can be transferred to (discrete) congestion games
of Rosenthal~\cite{Rosenthal73a} as follows: If we compare two special
pure Nash equilibria (namely the global potential
minima) of two congestion models (for which only cost functions or demands are decreased),
then our characterizations still hold. For atomic splittable congestion games (cf.~Bhaskar et al.~\cite{Bhaskar15}, Correa et al.~\cite{CCS06}, Harks~\cite{Harks:stack2011} and Schoppmann and Roughgarden~\cite{Roughgarden15}) it seems unclear whether or not  similar results hold true.

\subsection*{Acknowledgments.}
We thank Christian Fersch and Andreas S. Schulz for helpful comments
of an earlier draft of the manuscript.
We thank two anonymous reviewers for their very valuable feedback.
In particular, we are grateful to one reviewer who pointed out the special role of clutters
in our characterization.
The work of the first author is supported 
by JSPS Grant-in-Aid for Scientific Research (B) 25280004.

\bibliographystyle{abbrv}
\bibliography{master-bib}

\end{document}